%% file: main.tex
\documentclass{article}

\input{__arxiv_preamble}

\usepackage[noend]{algpseudocode}
 \usepackage{algorithmicx}
    \usepackage[ruled]{algorithm}
    \usepackage{algpseudocode}

  \usepackage{amsmath,amsfonts}
  \usepackage{tikz}
  \usetikzlibrary{trees}
\usetikzlibrary{decorations.pathmorphing}
\usetikzlibrary{decorations.markings}
\usetikzlibrary{decorations.pathmorphing,shapes}
\usetikzlibrary{calc,decorations.pathmorphing,shapes}
  \usepackage[T1]{fontenc}
  \usepackage[utf8]{inputenc}
  \usepackage{comment}
  \usepackage{forest}
  \usepackage{xspace}
  \usepackage{enumerate}
  \usepackage{listings}
  \usepackage{multirow}
  \usepackage{todonotes}
  \usepackage{thmtools}
  \usepackage{thm-restate}
  \usepackage{hyperref}
  \usepackage[capitalise]{cleveref}
  \usepackage{graphics,adjustbox}
  \usepackage{tabularx}
  \usepackage{amssymb}
  \usepackage{epsfig}
  \usepackage[whileod]{chl-alg}
  \usepackage[caption=true,font=footnotesize]{subfig}
  \usepackage{verbatim}
  \usepackage{url}
  \usetikzlibrary{trees, arrows, shapes}
\tikzset{
  treenode/.style = {align=center, inner sep=2pt, text centered,
    font=\sffamily},
  arn_r/.style = {treenode, circle, black, font=\sffamily\bfseries, draw=black,
    text width=1.5em},
    arn_t/.style = {treenode, circle, black, thick, double, font=\sffamily\bfseries, draw=black,
    text width=1.5em},
  every edge/.append style={anchor=south,auto=falseanchor=south,auto=false,font=3.5 em},
}

\def\dd{\mathinner{.\,.}}
\newcommand{\cO}{\mathcal{O}}
\newcommand{\st}{\mathcal{ST}}
\newcommand{\tr}{\mathcal{T}}
\newcommand{\SA}{\textsf{SA}}
\newcommand{\LCP}{\textsf{LCP}}

\newcommand{\lcp}{\textsf{lcp}}
\newcommand{\lcs}{\textsf{lcs}}
\newcommand{\rev}{\textsf{rev}}
\newcommand{\iSA}{\textsf{iSA}}

\newcommand{\range}{\textit{range}}

 \newcommand{\defproblem}[3]{
  \vspace{2mm}
\noindent\fbox{
  \begin{minipage}{0.96\textwidth}
  #1\\
  {\bf{Input:}} #2  \\
  {\bf{Output:}} #3
  \end{minipage}
  }
  \vspace{2mm}
}

\begin{document}
\maketitle

\begin{abstract}
In the $k$-mappability problem, we are given a string $x$ of length $n$ and integers $m$ and $k$, and we are asked to count, for each length-$m$ factor $y$ of $x$, the number of other factors of length $m$ of $x$ that are at Hamming distance at most $k$ from $y$. We focus here on the version of the problem where $k=1$. The fastest known algorithm for $k=1$ requires time $\mathcal{O}(mn \log n / \log\log n)$ and space $\mathcal{O}(n)$.  We present two algorithms that require worst-case time $\mathcal{O}(mn)$ and $\cO(n \log^2 n)$, respectively, and space $\mathcal{O}(n)$, thus greatly improving the state of the art. Moreover, we present an algorithm that requires average-case time and space $\mathcal{O}(n)$ for integer alphabets if $m=\Omega (\log n / \log \sigma)$, where $\sigma$ is the alphabet size.
\end{abstract}

\section{Introduction}

The focus of this work is directly motivated by the well-known and challenging application of {\em genome re-sequencing}---the assembly of a genome directed by a reference sequence. New developments in sequencing technologies~\cite{Metzker} allow whole-genome sequencing to be turned into a routine procedure, creating sequencing data in massive amounts. Short sequences, known as {\em reads}, are produced in huge amounts (tens of gigabytes); and in order to determine the part of the genome from which a read was derived, it must be mapped (aligned) back to some reference sequence that consists of a few gigabases.
A wide variety of short-read alignment techniques and tools have been published in the past years to address the challenge of efficiently mapping tens of millions of reads to a genome, focusing on different aspects of the procedure: speed, sensitivity, and accuracy~\cite{Fonseca}. These tools allow for a small number of errors in the alignment.

The \emph{$k$-mappability} problem was first introduced in the context of genome analysis in~\cite{biopaper} (and in some sense earlier in~\cite{ITAB2009}), where a heuristic algorithm was proposed to approximate the solution. The aim from a biological perspective is to compute the mappability of each region of a genome sequence; i.e.~for every factor of a given length of the sequence, we are asked to count how many other times it occurs in the genome with up to a given number of errors. This is particularly useful in the application of genome re-sequencing. By computing the mappability of the reference genome, we can then assemble the genome of an individual with greater confidence by first mapping the segments of the DNA that correspond to regions with low mappability. Interestingly, it has been shown that genome mappability varies greatly between species and gene classes~\cite{biopaper}.

Formally, we are given a string $x$ of length $n$ and integers $m<n$ and $k<m$, and we are asked to count, for each length-$m$ factor $y$ of $x$, the number of other length-$m$ factors of $x$ that are at Hamming distance at most $k$ from $y$.

\begin{example}
  Consider the string $x=\texttt{aabaaabbbb}$ and $m=3$.
  The following table shows the $k$-mappability counts for $k=0$ and $k=1$.
  \begin{center}
    \begin{tabular}{r|c|c|c|c|c|c|c|c}
      \textbf{position} & 0 & 1 & 2 & 3 & 4 & 5 & 6 & 7 \\\hline
      \textbf{factor occurrence} & $\texttt{aab}$ & $\texttt{aba}$ & $\texttt{baa}$ & $\texttt{aaa}$ & $\texttt{aab}$ & $\texttt{abb}$ & $\texttt{bbb}$ & $\texttt{bbb}$ \\\hline
      \textbf{0-mappability} & 1 & 0 & 0 & 0 & 1 & 0 & 1 & 1\\\hline
      \textbf{1-mappability} & 3 & 2 & 1 & 4 & 3 & 5 & 2 & 2
    \end{tabular}
  \end{center}
For instance, consider the position 0. The 0-mappability is 1, as the factor $\texttt{aab}$ occurs also at position 4. The 1-mappability at this position is 3 due to the occurrence of $\texttt{aab}$ at position 4 and occurrences of two factors at Hamming distance $1$ from $\texttt{aab}$: $\texttt{aaa}$ at position $3$ and $\texttt{abb}$ at position 5.
\end{example}

The $0$-mappability problem can be solved in $\mathcal{O}(n)$ time with the well-known LCP data structure~\cite{indLCP}. For $k=1$, to the best of our knowledge, the fastest known algorithm is by Manzini~\cite{Manzini:2015:LCP:2952649.2952678}. This solution runs in $\mathcal{O}(mn \log n / \log\log n)$ time and $\mathcal{O}(n)$ space and works only for strings over a fixed-sized alphabet. Since the problem for $k=0$ can be solved in $\mathcal{O}(n)$ time, one may focus on counting, for each length-$m$ factor $y$ of $x$, the number of other factors of $x$ that are at Hamming distance {\em exactly} $1$ --- instead of at most $1$ --- from $y$.

\medskip
\noindent \textbf{Our contributions.} Here we make the following threefold contribution:%
\begin{description}
  \item[(a)] We present an algorithm that, given a string of length $n$ over a fixed-sized alphabet and a positive integer $m$, solves the $1$-mappability problem in $\mathcal{O}(\min \{mn, n\log^2 n \})$ time and $\mathcal{O}(n)$ space, thus improving on the algorithm of~\cite{Manzini:2015:LCP:2952649.2952678} that requires $\mathcal{O}(mn \log n / \log\log n)$ time and $\mathcal{O}(n)$ space.
  \item[(b)] We present an algorithm to solve the $1$-mappability problem in $\mathcal{O}(mn)$ time and $\mathcal{O}(n)$ space that works for strings over an integer alphabet.
  \item[(c)] We present an algorithm that, given a string $x$ of length $n$ over an integer alphabet $\Sigma$ of size $\sigma>1$, with the letters of $x$ being independent and identically distributed random variables, uniformly distributed over $\Sigma$, and a positive integer $m=\Omega \Big(\frac{\log n}{\log \sigma}\Big)$, solves the $1$-mappability problem for $x$ in average-case time $\mathcal{O}(n)$ and space $\mathcal{O}(n)$. 
\end{description}

The paper is organised as follows. In Section~\ref{sec:prel}, we provide basic definitions and notation as well as a description of the algorithmic tools we use to design our algorithms. In Sections~\ref{sec:fast} and~\ref{sec:worst}, we provide the average-case and the worst-case algorithms, respectively. We conclude with some final remarks in Section~\ref{sec:conc}.

\section{Preliminaries}\label{sec:prel}
  We begin with some basic definitions and notation.
  Let $x=x[0]x[1]\ldots x[n-1]$ be a \textit{string} of length $|x|=n$ over a finite ordered alphabet $\Sigma$ of size $|\Sigma|=\sigma=\cO(1)$. We also consider the case of strings over an {\em integer alphabet}, where each letter is replaced by its rank in such a way that the resulting string consists of integers in the range $\{1,\ldots,n\}$. 

For two positions $i$ and $j$ on $x$, we denote by $x[i\dd j]=x[i]\ldots x[j]$ the \textit{factor} 
(sometimes called \textit{substring}) of $x$ that 
starts at position $i$ and ends at position $j$ (it is of length $0$ if $j<i$). By $\varepsilon$ we denote
the \textit{empty string} of length 0. 
  We recall that a prefix of $x$ is a factor that starts at position 0 
($x[0\dd j]$) and a suffix of $x$ is a factor that ends at position $n-1$ 
($x[i\dd n-1]$). We denote the reverse string of $x$ by $\textsf{rev}(x)$, i.e.~$\rev(x)=x[n-1]x[n-2]\ldots x[1]x[0]$.

  Let $y$ be a string of length $m$ with $0<m\leq n$. 
  We say that there exists an \textit{occurrence} of $y$ in $x$, or, more 
simply, that $y$ \textit{occurs in} $x$, when $y$ is a factor of $x$.
  Every occurrence of $y$ can be characterised by a starting position in $x$. 
  Thus we say that $y$ occurs at the \textit{starting position} $i$ in $x$ when $y=x[i \dd i + m - 1]$.
  
  The {\em Hamming distance} between two strings $x$ and $y$ of the same length is defined as $d_H(x, y) = |\{i : x[i] \neq y[i],\, i = 0, 1,\ldots, |x| - 1\}|$. If $|x| \neq |y|$, we set $d_H(x, y)=\infty$. If two strings $x$ and $y$ are at Hamming distance $k$, we write $x \approx_k y$. 
  
  The computational problem in scope can be formally stated as follows.

{\defproblem{\textsc{1-mappability}}{A string $x$ of length $n$ and an integer $m$, $1 \leq m <n$}
{An integer array $C$ of size $n-m+1$ such that $C[i]$ stores the number of factors of $x$ that are at Hamming distance $1$ from $x[i\dd i+m-1]$}}

\subsection{Suffix array and suffix tree}
Let $x$ be a string of length $n>0$. We denote by \SA{} the {\em suffix array} of $x$. \SA{} is an integer array of size $n$ storing the starting positions of all (lexicographically) sorted non-empty suffixes of $x$, i.e.~for all 
$1 \leq  r < n$ we have $x[\SA{}[r-1] \dd n-1] < x[\SA{}[r] \dd n - 1]$~\cite{SA}.
  Let \lcp{}$(r, s)$ denote the length of the longest common prefix between
$x[\SA{}[r] \dd n - 1]$ and $x[\SA{}[s] \dd n - 1]$ 
for positions $r$, $s$ on $x$.
  We denote by \LCP{} the {\em longest common prefix} array of $x$ defined by 
\LCP{}$[r]=\lcp{}(r-1, r)$ for all $1 \leq r < n$, and 
\LCP{}$[0] = 0$. The inverse \iSA{} of the array \SA{} is defined by 
$\iSA{}[\SA{}[r]] = r$, for all $0 \leq r < n$. It is known that
  \SA{}, \iSA{}, and \LCP{} of a string of length $n$, over an integer alphabet, can be computed in time and space $\cO(n)$~\cite{Nong:2009:LSA:1545013.1545570,indLCP}. It is then known that a range minimum query (RMQ) data structure over the \LCP{} array, that can be constructed in $\cO(n)$ time and $\cO(n)$ space~\cite{Bender2000}, can answer \lcp{}-queries in $\cO(1)$ time per query~\cite{SA}. A symmetric construction on $\rev(x)$ can answer the so-called \emph{longest common suffix} (\lcs) queries in the same complexity. The \lcp{} and \lcs{} queries are also known as {\em longest common extension} (LCE) queries.

The \textit{suffix tree} $\mathcal{T}(x)$ of string $x$ is a compact trie representing all suffixes of $x$. The nodes of the trie which become nodes of the suffix
tree are called {\em explicit} nodes, while the other nodes are called {\em implicit}. Each edge
of the suffix tree can be viewed as an upward maximal path of implicit nodes starting with an explicit node. Moreover, each node belongs to a unique path of that kind. Thus, each node of the trie can be represented in the suffix tree by the edge it belongs to and an index within the corresponding path. The label of an edge is its first letter.
We let  $\mathcal{L}(v)$  denote the \textit{path-label} of a node $v$, i.e., the concatenation of the edge labels along the path from the root to $v$. We say that $v$ is  path-labelled  $\mathcal{L}(v)$. Additionally, $\mathcal{D}(v)= |\mathcal{L}(v)|$ is used to denote  the \textit{string-depth} of node $v$. Node  $v$ is  a \textit{terminal} node if its path-label is a suffix of $x$, that is, $\mathcal{L}(v) = x[i \dd n-1]$ for some $0 \leq i < n$; here $v$ is also labelled with index $i$. It should be clear that each  factor of $x$ is uniquely represented by either an explicit or an implicit node of $\mathcal{T}(x)$. 
In standard suffix tree implementations, we assume that each node of the suffix tree is able to access its parent. Once $\mathcal{T}(x)$ is constructed, it can be traversed in a depth-first manner to compute $\mathcal{D}(v)$ for each node $v$. 

It is known that the suffix tree of a string of length $n$, over an integer alphabet, can be computed in time and space $\cO(n)$~\cite{farach1997optimal}.
For integer alphabets, in order to access the children of an explicit node by the first letter of their edge label, perfect hashing~\cite{DBLP:journals/jacm/FredmanKS84} can be used.

\section{Efficient Average-Case Algorithm}\label{sec:fast}
In this section we assume that $x$ is a string over an integer alphabet $\Sigma$.
Recall that if two strings $y$ and $z$ are at Hamming distance $1$, we write $y \approx_1 z$. 
\begin{fact}[Folklore]\label{fact:pp}
Given two strings $y$ and $z$ of length $m$, we have that if \mbox{$y \approx_1 z$}, then $y$ and $z$ share at least one factor of length $\lfloor m/2 \rfloor$.
\end{fact}

\begin{fact}\label{fact:m3}
Given a string $x$ and any two positions $i,j$ on $x$, we have that if $x[i \dd i+m-1] \approx_1 x[j \dd j+m-1]$, then $x[i \dd i+m-1]$ and $x[j \dd j+m-1]$ have at least one common factor of length $L = \lfloor m/3 \rfloor$ starting at positions $i' \in \{i,\ldots,i+m-L\}$ and $j' \in \{j,\ldots,j+m-L\}$ of $x$, such that $i'-i=j'-j$ and $i'= 0 \pmod L$.
\end{fact}

\begin{proof}
It should be clear that every factor of $x$ of length $m$ fully contains at least two factors of length $L$ starting at positions equal to $0$ mod $L$. Then, if $x[i \dd i+m-1]$ and $x[j \dd j+m-1]$ are at Hamming distance $1$, analogously to Fact~\ref{fact:pp}, at least one of the two factors of length $L$ that are fully contained in $x[i \dd i+m-1]$ occurs at a corresponding position in $x[j \dd j+m-1]$; otherwise we would have a Hamming distance greater than 1.
\end{proof}

We first initialize an array $C$ of size $n-m+1$, with $0$ in all positions; for all $i$, $C[i]$ will eventually store the number of factors of $x$ that are at Hamming distance $1$ from $x[i \dd i+m-1]$.
We apply Fact~\ref{fact:m3} by implicitly splitting the string $x$ into $B = {\lfloor \tfrac{n}{\lfloor m/3 \rfloor}\rfloor}$ {\em blocks} of length $L=\lfloor m/3 \rfloor$---the suffix of length $n \bmod \lfloor m/3 \rfloor$ is not taken as a block---starting at the positions of $x$ that are equal to $0$ mod $L$.
In order to find all pairs of length-$m$ factors that are at Hamming distance $1$ from each other, we can find all the exact matches of every block and try to extend each of them to the left and to the right, allowing at most one mismatch. 
However, we need to tackle some technical details to correctly update our counters and avoid double counting.

We start by constructing the \SA{} and \LCP{} arrays for $x$ and $\rev(x)$ in $\cO(n)$ time. 
We also construct RMQ data structures over the \LCP{} arrays for answering LCE queries in constant time per query. By exploiting the \LCP{} array information, we can then find in $\cO(n)$ time all maximal sets of indices such that the longest common prefix between any two of the suffixes starting at these indices is at least $L$ and at least one of them is the starting position of some block. 

Then for each such set, denoted by $P$, we have to do the following procedure for each index $i \in P$ such that $i=0 \pmod L$.

For every other $j \in P$, we try to extend the match by asking two LCE queries in each direction. I.e., we ask an $\lcs(i-1,j-1)$ query to find the first mismatch positions $\ell_1$ and $\ell'_1$, respectively, and then $\lcs(\ell_1-1,\ell'_1-1)$ to find the second mismatch ($\ell_2$ and $\ell'_2$, respectively). A symmetric procedure computes the mismatches $r_1,r'_1$ and $r_2,r'_2$ to the right, as shown in Figure~\ref{lces}. We omit here some technical details with regards to reaching the start or end of $x$.

  \begin{figure}
  \begin{center}
  \begin{tikzpicture}[xscale=0.6]
    \foreach \x/\c in {-0.7/\ell_2,0/,1/p,2/,3.1/\ell_1,4/,5.2/q,6/,7.5/i,8/,9/,10/,11/i+L-1,12/,13/,14.1/r_1,15/,16/,17/,18/, 18.7/r_2}{
      \draw (\x,0) node[above] {$\c$};
    }
    \draw[latex-,decorate,decoration={snake,amplitude=.4mm,segment length=2mm}] (3.8,0.7) -- (6.8,0.7);
    \draw[-latex,decorate,decoration={snake,amplitude=.4mm,segment length=2mm}] (10.5,0.7) -- (13.3,0.7);
    \draw[latex-,decorate,decoration={snake,amplitude=.4mm,segment length=2mm}] (0,0.7) -- (2.4,0.7);
    \draw[-latex,decorate,decoration={snake,amplitude=.4mm,segment length=2mm}] (14.4,0.7) -- (18,0.7);
    \draw[latex-] (7.2,0.7) -- (9.7,0.7);
    \draw[-latex] (9.7,0.7) -- (10.1,0.7);
    \draw (3.1,0.45) node[above] {$X$};
    \draw (14,0.45) node[above] {$X$};
    \draw (-0.7,0.45) node[above] {$X$};
    \draw (18.7,0.45) node[above] {$X$};
    
\begin{scope}[yshift=1.5cm]
    \foreach \x/\c in {-0.7/\ell'_2,0/,1/p',2/,3.1/\ell'_1,4/,5.2/q',6/,7.5/j,8/,9/,10/,11/j+L-1,12/,13/,14.1/r'_1,15/,16/,17/,18/, 18.7/r'_2}{
      \draw (\x,0) node[above] {$\c$};
    }
    \draw[latex-,decorate,decoration={snake,amplitude=.4mm,segment length=2mm}] (3.8,0.7) -- (6.8,0.7);
    \draw[-latex,decorate,decoration={snake,amplitude=.4mm,segment length=2mm}] (10.5,0.7) -- (13.3,0.7);
    \draw[latex-,decorate,decoration={snake,amplitude=.4mm,segment length=2mm}] (0,0.7) -- (2.4,0.7);
    \draw[-latex,decorate,decoration={snake,amplitude=.4mm,segment length=2mm}] (14.4,0.7) -- (18,0.7);
    \draw[latex-] (7.2,0.7) -- (9.7,0.7);
    \draw[-latex] (9.7,0.7) -- (10.1,0.7);
    \draw (3.1,0.45) node[above] {$X$};
    \draw (14,0.45) node[above] {$X$};
    \draw (-0.7,0.45) node[above] {$X$};
    \draw (18.7,0.45) node[above] {$X$};
\end{scope}
  \end{tikzpicture}
  \end{center}
  \caption{Performing two LCE queries in each direction.}
  \label{lces}
  \end{figure}
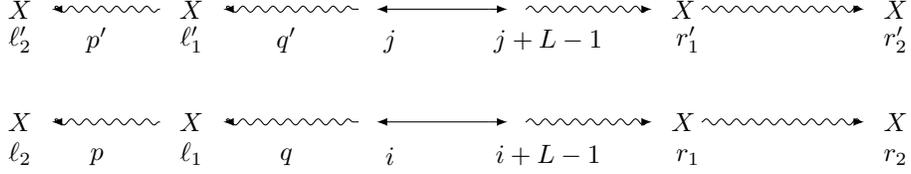


Now we are interested in positions $p$ such that $\ell_2 < p \leq \ell_1$ and $i+ L -1 \leq p+m-1 < r_1$ and positions $q$ such that $\ell_1 < q \leq i$ and $r_1 \leq q+m-1 < r_2$. Each such position $p$ (resp. $q$) implies that $x[p\dd p+m-1] \approx_1 x[p'\dd p'+m-1]$, where $p'=j-(i-p)$. Henceforth, we only consider positions of the type $p,p'$.

Note that if $x[p\dd p+m-1] \approx_1 x[p'\dd p'+m-1]$, we will identify the unordered pair $\{p,p'\}$ based on the described approach $t_{p, p'}$ times, where $t_{p,p'}$ is the total number of full blocks contained in $x[p \dd p+m-1]$ and in $x[p' \dd p'+m-1]$ after the mismatch position. It is not hard to compute the number $t_{p, p'}$ in $\cO(1)$ time based on the starting positions $p$ and $p'$ as well as $\ell_1$ and $r_1$ each time we identify $x[p\dd p+m-1] \approx_1 x[p'\dd p'+m-1]$. To avoid double counting, we then increment the $C[p]$ and $C[p']$ counters by $1/t_{p, p'}$.

By $\textsf{EXT}_{i,j}$ we denote the time required to process a pair of elements $i,j$ of a set $P$ such that at least one of them, $i$ or $j$, equals $0$ mod $L$. 
\begin{lemma}\label{lem:ext}
The time $\textsf{EXT}_{i,j}$ is $\mathcal{O}(m)$.
\end{lemma}
\begin{proof}
Given $i,j \in P$, with at least one of them equal to $0$ mod $L$, we can find the pairs $(p,p')$ of positions that satisfy the inequalities discussed above in $\cO(m)$ time. They are a subset of $\{(i-m+L,j-m+L), \ldots ,(i-1,j-1)\}$. For each such pair $(p,p')$ we can compute $t_{p, p'}$ and increment $C[p]$ and $C[p']$ accordingly in $\cO(1)$ time. The total time to process all pairs $(p,p')$ for given $i,j$ is thus $\mathcal{O}(m)$. 
\end{proof}

\begin{theorem}\label{the:1map}
Given a string $x$ of length $n$ over an integer alphabet $\Sigma$ of size $\sigma>1$ with the letters of $x$ being independent and identically distributed random variables, uniformly distributed over $\Sigma$, the \textsc{1-mappability} problem can be solved in average-case time $\mathcal{O}(n)$ and space $\mathcal{O}(n)$ if $m \geq 3 \cdot \frac{\log n}{\log \sigma}+3$. 
\end{theorem}

\begin{proof}
The time and space required for constructing \SA{} and \LCP{} tables for $x$ and $\rev(x)$ and the RMQ data structures over the \LCP{} tables is $\mathcal{O}(n)$.

Let $B$ denote the number of blocks over $x$, and let $L$ denote the block length. We set
$$\quad L={\lfloor \tfrac{m}{3} \rfloor}, \quad B = {\lfloor \tfrac{n}{L}\rfloor}$$
\noindent to apply Fact~\ref{fact:m3}. Recall that by $P$ we denote a maximal set of indices of the \LCP{} table such that the length of the longest common prefix between any two suffixes starting at these indices is at least $L$ and at least one of them is the starting position of some block. Processing all such sets $P$ requires time
$$ \textsf{EXT}_{i,j} \cdot \textit{Occ}$$
\noindent where $\textsf{EXT}_{i,j}$ is the time required to process a pair $i,j$ of elements of a set $P$; and \textit{Occ} is the sum of the multiples of the cardinality of each set $P$ times the number of the elements of set $P$ that are equal to $0$ mod $L$. By Lemma~\ref{lem:ext} we have that $ \textsf{EXT}_{i,j}=\mathcal{O}(m)$. Additionally, by the stated assumption on the string $x$, the expected value for \textit{Occ} is no more than $\frac{Bn}{\sigma^L}$. Hence, the algorithm on average requires time
$$\mathcal{O}(n + m \cdot \frac{B \cdot n}{\sigma^{L}}).$$
Assuming that $m > 3$, we have the following:

$$m \cdot \frac{B \cdot n}{\sigma^{L}} = 
{\frac {m\cdot \lfloor \tfrac{n}{\lfloor m/3 \rfloor}\rfloor \cdot n }{\sigma ^ {\lfloor\frac{m}{3}\rfloor}}}
\leq {\frac {m\cdot  (\tfrac{n}{ m/3-1 }) \cdot n}{\sigma ^ {\frac{m}{3}-1}}}
\leq {\frac {12n^2}{n ^ {\frac{\log \sigma}{\log n} (\frac{m}{3}-1)}}}
= 12n^{2-\frac{(m-3) \log \sigma}{3 \log n}}.$$

\noindent Consequently, in the case when $$m \geq 3 \cdot \frac{\log n}{\log \sigma}+3$$ the algorithm requires $\mathcal{O}(n)$ time on average. The extra space usage is $\cO(n)$.
\end{proof}

\section{Efficient Worst-Case Algorithms}\label{sec:worst}

\subsection{$\cO(mn)$-time and $\cO(n)$-space algorithm}\label{subsec:nm}

In this section we assume that $x$ is a string over an integer alphabet $\Sigma$.
The main idea is that we want to first find all pairs $x[i_1 \dd i_1+m-1] \approx_1 x[i_2 \dd i_2+m-1]$ that have a mismatch in the first position, then in the second, and so on.

Let us fix $0 \le j < m$. In order to identify the pairs $x[i_1 \dd i_1+m-1] \approx_1 x[i_2 \dd i_2+m-1]$ with $x[i_1+j] \neq x[i_2+j]$ (i.e. with the mismatch in the $j^{th}$ position), we do the following. 
For every $i=0,1, \ldots ,n-m$, we find the explicit or implicit node $u_{i,j}$ in $\tr(x)$ that represents $x[i \dd i+j-1]$ and the node $v_{i,j}$ in $\tr(\rev(x))$ that represents $\rev(x[i+j+1 \dd i+m-1])=\rev(x)[n-i-m \dd n-i-j-2]$.
In each such node $v_{i,j}$, we create a set $V(v_{i,j})$---if it has not already been created---and insert the triple $(u_{i,j},x[i+j],i)$.

When we have done this for all possible starting positions of $x$, we group the triples in each set $V(v)$ by the node variable (i.e., the first component in the triples). For each such group in $V(v)$ we count the number of triples that have each letter of the alphabet and increment array $C$ accordingly.
More precisely, if $V(v)$ contains $q$ triples that correspond to the same node $u$, among which $r$ correspond to the letter $c \in \Sigma$, then for each such triple $(u,c,i) \in V(v)$ we increment $C[i]$ by $q-r$; we subtract $r$ to avoid counting equal factors in $C$.
Before we proceed with the computations for the next index $j$, we delete all the sets $V(v)$.
We formalize this algorithm, denoted by \textsc{$1$-Map}, in the pseudocode presented below and provide an example.

\bigskip
\begin{algo}{1-Map}{x,n,m}
\SET{\tr(x)}{\Call{SuffixTree}{x}}
\SET{\tr(\rev(x))}{\Call{SuffixTree}{\rev(x)}}
\DOFOR{\mbox{string-depth $j=0$ {\bf to} $m-1$}}
	\DOFOR{\mbox{$i=0$ {\bf to} $n-m$}}
    	\SET{u_{i,j}}{\Call{Node$_{\tr(x)}$}{x[i\dd i+j-1]}}
    	\SET{v_{i,j}}{\Call{Node$_{\tr(\rev(x))}$}{\rev(x)[n-i-m \dd n-i-j-2]}}
        \ACT{\mbox{Insert $(u_{i,j}, x[i+j], i)$ to $V(v_{i,j})$}}
     \OD
     \DOFOR{\mbox{every node $v$ of string-depth $m-j-2$ in $\tr(\rev(x))$}}
       	\ACT{\mbox{Group triples in $V(v)$ by the node variable}}
       	\DOFOR{\mbox{a group corresponding to the node $u$ in $V(v)$}}
        	\ACT{\mbox{Count number of triples with each letter $c \in \Sigma$}}
       		\ACT{\mbox{Update $C[i]$ accordingly for each triple $(u, c, i)$}}
        \OD
        \ACT{\mbox{Delete $V(v)$}}
     \OD
\OD
\end{algo}

\bigskip
\begin{example}
Suppose we have $V(v)=\{ (u, \texttt{A}, i_1), (u, \texttt{A}, i_2), (u, \texttt{A}, i_3), (u, \texttt{C}, i_4),(u, \texttt{C}, i_5), (u, \texttt{C}, i_6), (u, \texttt{G}, i_7),$\linebreak $(u, \texttt{G}, i_8), (u, \texttt{T}, i_9) \}$, for some distinct positions $i_1,i_2,$ $\ldots,i_9$. We then increment $C[i_1]$, $C[i_2]$, $C[i_3]$, $C[i_4]$, $C[i_5]$, and $C[i_6]$ by $6$; $C[i_7]$ and $C[i_8]$ by $7$; and $C[i_9]$ by $8$.
\end{example}

We now analyze the time complexity of this algorithm. The algorithm iterates $j$ from $0$ to $m-1$. 
In the $j^{th}$ iteration, we need to compute $\{ u_{i,j}, v_{i,j} \mid i = 0, \ldots, n-m\}$.
When $j=0$, $u_{i,0}$ for every $i$ is the root of $\tr(x)$ and we can find $v_{i,0}$ for all $i$ na\"ively in $\cO(mn)$ time. 
For $j>0$, $v_{i,j}$ can be found in $\cO(1)$ time from $v_{i,j-1}$ by moving one letter up in ${\cal T}(\mathtt{rev}(x))$ for all $i$, while $u_{i,j}$ can be obtained from $u_{i,j-1}$ by going down in ${\cal T}(x)$ based on letter $x[i+j]$. We then include $(u_{i,j},x[i+j],i)$ in $V(v_{i,j})$. 


This requires in total $\cO(mn)$ randomized time due to perfect hashing~\cite{DBLP:journals/jacm/FredmanKS84} which allows to go down from a node in ${\cal T}(x)$ (or in ${\cal T}(\rev(x))$) based on a letter in $\cO(1)$ randomized time. We can actually avoid this randomization, as queries for a particular child of a node are asked in our solution in a somewhat off-line fashion: we use them only to compute $v_{i,0}$ ($m$ times) and $u_{i,j}$ (from $u_{i,j-1}$).

\begin{observation}
For an integer alphabet $\Sigma=\{1,\ldots,n\}$, one can answer off-line $\cO(n)$ queries in $\tr(x)$ asking for a child of an explicit or implicit node $u$ labelled with the letter $c \in \Sigma$ in (deterministic) $\cO(n)$ time.
\end{observation}
\begin{proof}
A query for an implicit node $u$ is answered in $\cO(1)$ time, as there is only one outgoing edge to check. All the remaining queries can be sorted lexicographically as pairs $(u,c)$ using radix sort. We can also assume that the children of every explicit node of $\tr(x)$ are ordered by the letter (otherwise we also radix sort them). Finally, all the queries related to a node $u$ can be answered in one go by iterating through the children list of $u$ once.
\end{proof}

Lastly, we use bucket sort to group the triples for each $V(v)$ according to the node variable (recall that the nodes are represented by the edge and the index within the edge) and update the counters in $\cO(n)$ time in total (using a global array indexed by the letters from $\Sigma$, which is zeroed in $\cO(|V(v)|)$ time after each $V(v)$ has been processed). Overall the algorithm requires $\cO(mn)$ time.

The suffix trees require $\cO(n)$ space and we delete the sets $V(v_{i,j})$ after the $j^{th}$ iteration; the space complexity of the algorithm is thus $\cO(n)$. We obtain the following result.


\begin{theorem}\label{worstcase1}
Given a string $x$ of length $n$ over an integer alphabet and a positive integer $m$, we can solve the \textsc{1-mappability} problem in $\mathcal{O}(mn)$ time and $\mathcal{O}(n)$ space.
\end{theorem}

\subsection{$\cO(n \log^2 n)$-time and $\cO(n)$-space algorithm}

In this section we assume that $x$ is a length-$n$ string over an ordered alphabet $\Sigma$, where $|\Sigma|=\sigma=\cO(1)$. Consider two factors of $x$ represented by nodes $u$ and $v$ in $\tr (x)$; the first observation we make is that the first mismatch between the two factors is the first letter of the labels of the distinct outgoing edges from the lowest common ancestor of $u$ and $v$ 
that lie on the paths from the root to $u$ and $v$. For $1$-mappability we require that what follows this mismatch is an exact match.

\begin{definition}
Let $T$ be a rooted tree. For each non-leaf node $u$ of $T$, the {\em heavy edge} $(u,v)$ is an edge for which the subtree rooted at $v$ has the maximal number of leaves (in case of several such subtrees, we fix one of them). The {\em heavy path of a node $v$} is a maximal path of heavy edges that passes through $v$ (it may contain 0 edges). The {\em heavy path of $T$} is the heavy path of the root of $T$.
\end{definition}

Consider the suffix tree $\tr(x)$ and its node $u$. We say that an (explicit or implicit) node $v$ is a \emph{level ancestor of $u$ at string-depth $\ell$} if $\mathcal{D}(v)=\ell$ and $\mathcal{L}(v)$ is a prefix of $\mathcal{L}(u)$. The heavy paths of $\tr(x)$ can be used to compute level ancestors of nodes in $\cO(\log n)$ time. However, a more efficient data structure is known.

\begin{lemma}[\cite{DBLP:journals/talg/AmirLLS07}]\label{lem:LAQ}
  After $\cO(n)$-time preprocessing on $\tr(x)$, level ancestor queries of nodes of $\tr(x)$ can be answered in $\cO(\log \log n)$ time per query.
\end{lemma}

\tikzset{
itria/.style={
  draw,shape border uses incircle,
  isosceles triangle,shape border rotate=90, xshift=-0.2cm,yshift=-0.25cm},
rtria/.style={
  draw,shape border uses incircle,
  isosceles triangle,isosceles triangle apex angle=90,
  shape border rotate=-45,xshift=-0.2cm,yshift=-0.25cm},
ritria/.style={
  draw,shape border uses incircle,
  isosceles triangle,isosceles triangle apex angle=110,
  shape border rotate=-55,xshift=-0.2cm,yshift=-0.25cm},
letria/.style={
  draw,shape border uses incircle,
  isosceles triangle,isosceles triangle apex angle=90,
  shape border rotate=225,xshift=-0.32cm,yshift=-0.25cm}
}
\begin{figure}[!t]
\begin{center}
\begin{tikzpicture}

\node[shape=circle,draw=black] (3) at (9,-3) {};  

\node[shape=circle,draw=white, label = {left: $c$}] (ccc) at (8.8,-4) {};
\node[shape=circle,draw=white, label = {left: $d$}] (ddd) at (7.8,-4.2) {};
\node[shape=circle,draw=white, label = {left: $v$}] (vvv) at (9.7,-4.27) {};
\node[shape=circle,draw=white, label = {left: $u$}] (uuu) at (8.1,-3.8) {};
\node[shape=circle,draw=white, label = {left: $u'$}] (u'u'u') at (7.1,-4.8) {};
\node[shape=circle,draw=white, label = {left: $z$}] (zzz) at (9.85,-4.7) {};
\node[shape=circle,draw=white, label = {left: $z'$}] (z'z'z') at (8.9,-5.7) {};

\node[shape=circle,draw=black] (8a) at (8,-4) {};
\node[shape=circle,draw=white] (8b) at (9.5,-3.5) {};
\node[letria,draw=black] (8btree) at (9.85,-3.5) {};
\node[shape=circle,draw=white] (8c) at (10.5,-3.44) {};
\node[letria,draw=black] (8ctree) at (10.83,-3.5) {};
\node[shape=circle,draw=white] (8d) at (11.5,-3.42) {};
\node[letria,draw=black] (8dtree) at (11.83,-3.5) {};
\draw [->,decorate,draw=red] (3) -- (8a);
\draw [->,decorate] (3) -- (8b);
\draw [->,decorate] (3) -- (8c);
\draw [->,decorate] (3) -- (8d);


\node[shape=circle,draw=white] (9) at (8.5,-4.5) {};
\node[letria,draw=black] (9tree) at (8.85,-4.5) {};
\node[shape=circle,draw=black] (9a) at (7,-5) {};
\node[shape=circle,draw=white, label = {right: $S_i$}] (9b) at (9.5,-4.44) {};
\node[letria,draw=black] (9btree) at (9.83,-4.5) {};
\draw [->,decorate] (8a) -- (9);
\draw [->,decorate,draw=red] (8a) -- (9a);
\draw [->,decorate] (8a) -- (9b);

\node[shape=circle,draw=white] (10) at (7.5,-5.5) {};
\node[letria,draw=black] (10tree) at (7.85,-5.5) {};
\node[shape=circle,draw=black] (10a) at (5.5,-6.5) {};
\node[shape=circle,draw=white] (10b) at (8.5,-5.44) {};
\node[letria,draw=black] (10btree) at (8.83,-5.5) {};

\draw [->,decorate] (9a) -- (10);
\draw [->,decorate,decoration=snake,draw=red] (9a) -- (10a);
\draw [->,decorate] (9a) -- (10b);

\node[shape=circle,draw=white] (11) at (6,-7) {};
\node[letria,draw=black] (11tree) at (6.35,-7) {};
\node[shape=circle,draw=black] (11a) at (4.5,-7.5) {};
\node[circle,scale=0.6,draw=black] (term) at (4.5,-7.5) {};
\node[shape=circle,draw=white] (11b) at (7,-6.94) {};
\node[letria,draw=black] (11btree) at (7.33,-7) {};
\draw [->,decorate] (10a) -- (11);
\draw [->,decorate,draw=red] (10a) -- (11a);
\draw [->,decorate] (10a) -- (11b);
\end{tikzpicture}
\end{center}
\caption{Illustration; the heavy path of $\tr(x)$ is shown in red.}
\label{fig:algo}
\end{figure}
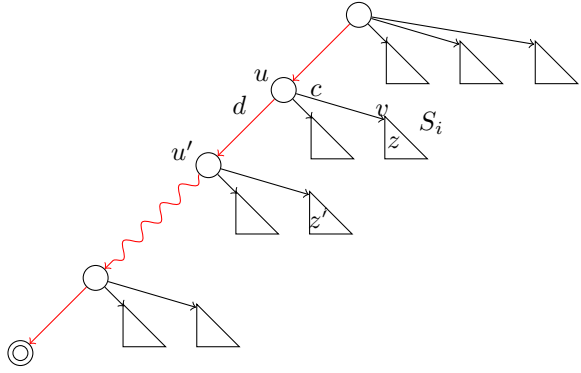

\vspace*{-0.1cm}
\begin{definition}
Given a string $x$ and a factor $y$ of $x$, we denote by $\range(x, y)$ the range in the \textsf{SA} of $x$ that represents the suffixes of $x$ that have $y$ as a prefix.
\end{definition}

Every node $u$ in $\tr(x)$ corresponds to an \textsf{SA} range $I_u = \range(x, \mathcal{L}(u))=(u_{\min}, u_{\max})$.
We can precompute $I_u$ for all explicit nodes $u$ in $\tr(x)$ in $\cO(n)$ time while performing a depth-first traversal of the tree as follows.
For a non-terminal node $v$ with children $u^1, \ldots ,u^q$, we set $v_{\min}=\min_i \{ u^i_{\min} \}$ and $v_{\max}=\max_i \{ u^i_{\max} \}$. If $v$ is a terminal node (with children $u^1, \ldots ,u^q$), representing the suffix $x[j \dd n-1]$, we set $v_{\min}=\textsf{iSA}[j]$ and $v_{\max}=\max \{ \textsf{iSA}[j] , \max_i \{ u^i_{\max} \} \}$. When a considered node $v$ is implicit, say along an edge $(p,q)$, then $I_v=I_q$.

Our algorithm relies heavily on the following auxiliary lemmas.

\begin{lemma}\label{suffix-of-node-u}
Consider a node $u$ in $\tr(x)$ with $p = \mathcal{L}(u)$.
Let $\textit{suf}(u, \ell)$ be the node $v$ such that $\mathcal{L}(v) = p[\ell \dd |p|-1]$.
Given the \textsf{SA} and the \textsf{iSA} of $x$, $v$ can be computed in $\cO(\log \log n)$ time after $\cO(n)$-time preprocessing.
\end{lemma}
\begin{proof}
The \SA{} range of the node $u$ is $I_u = (u_{\min}, u_{\max})$;
$u_{\min}$ corresponds to the suffix $x[\SA[u_{\min}] \dd n-1]$.
By removing the first $\ell$ letters, the suffix becomes
$x[\SA[u_{\min}]+\ell \dd n-1]$.
The corresponding \SA{} value is
$v_{\min} = \iSA[\SA[u_{\min}]+\ell]$.

Let $v_1$ be the node of $\tr(x)$ such that ${\cal L}(v_1)=x[\SA [v_{\min}]\dd n-1]$.
The sought node $v$ is the ancestor of $v_1$ located at string-depth $|p|-\ell$.
It can be computed in $\cO(\log \log n)$ time using the level ancestor data structure of Lemma~\ref{lem:LAQ}.
\end{proof}

\begin{lemma}\label{pcp}
Let $u$ and $v$ be two nodes in $\tr(x)$.
We denote ${\cal L}(u)$ by $p_1$ and ${\cal L}(v)$ by $p_2$.
We further denote by $\text{concat}(u, v)$ the node $w$ such that ${\cal L}(w) = p_1 p_2$.
Given the \textsf{SA} and the \textsf{iSA} of $x$, as well as $\range(x, p_1)$ and $\range(x, p_2)$, $w$ can be located in $\cO(\log n)$ time after $\cO(n)$-time preprocessing.
\end{lemma}

\begin{proof}
We can compute $\range(x, p_1 p_2)=(w_{\min}, w_{\max})$ in $\cO(\log n)$ time using the \textsf{SA} and the \textsf{iSA} of $x$~\cite{G97,Huynh:2006:ASM:1142859.1142877}; we can then locate $w$ in $\cO(\log \log n)$ time using the level ancestor data structure of Lemma~\ref{lem:LAQ}.
\end{proof}


We are now ready to present an algorithm for $1$-mappability that requires $\cO(n \log^2 n)$ time and $\cO(n)$ space.
The first step is to build $\tr(x)$.
We then make every node $u$ of string-depth $m$ explicit in $\tr(x)$ and initialize a counter $\textit{Count}(u)$ for it.
 For each explicit node $u$ in $\tr(x)$, the \textsf{SA} range $I_u = \range(x, {\cal L}(u))$ is also stored.
We also identify the node $v_{c}$ with path-label $c$ for each $c \in \Sigma$ in $\cO(\sigma)=\cO(1)$ time.
 
\medskip
\medskip
\begin{algo}{PerformCount}{T,m}
\SET{\textit{HP}}{\Call{HeavyPath}{T}}
\DOFOR{\mbox{each side-tree $S_i$ attached to a node $u$ on $\textit{HP}$ with $\mathcal{D}(u)<m$}}
	\ACT{\text{Let $(u,v)$ be the edge that connects $S_i$ to $\textit{HP}$}}
    \SET{c}{\text{the edge label of $(u, v)$}}  
    \SET{d}{\text{the edge label of the heavy edge $(u, u')$}} 
	\DOFOR{\mbox{each node $z$ in $S_i$ with $\mathcal{D}(z)=m$}}
        \SET{w}{\text{suf}(z, \mathcal{D}(u)+1)}
        \DOFOR{\mbox{each $c' \neq c$, label of an outgoing edge from $u$}}
            \SET{t}{\text{concat}(u, \text{concat}(v_{c'}, w))}
        	\SET{\textit{Count}(z)}{\textit{Count}(z) + |I_t|}
        \OD
        \SET{z'}{\text{concat}(u, \text{concat}(v_d, w))}
        \SET{\textit{Count}(z')}{\textit{Count}(z') + |I_z|}
	\OD
\CALL{PerformCount}{S_i,m-\mathcal{D}(u)}
\OD
\end{algo}
 
\medskip
\medskip
We then call \textsc{PerformCount}$(\tr(x),m)$, which does the following (inspect also the pseudocode above and Figure~\ref{fig:algo}). At first, a heavy path $\textit{HP}$ of $\tr(x)$ is computed. Initially, we want to identify the pairs of factors of $x$ of length $m$ at Hamming distance $1$ that have a mismatch in the labels of the edges outgoing from a node in $\textit{HP}$.
Given a node $u$ in $\textit{HP}$, with $\mathcal{L}(u)=p_1$, for every side tree $S_i$ attached to it (say by an edge with label $c \in \Sigma$), we find all nodes of $S_i$ with string-depth $m$. For every such node $z$, with path-label $p_1 c p_2$, we use Lemma~\ref{suffix-of-node-u} to obtain the node $w = \text{suf}(z, |p_1|+1)$; that is, ${\cal L}(w)=p_2$. We then use Lemma~\ref{pcp} to compute $\range(x, p_1 c' p_2)$ for all $c' \neq c$ such that there is an outgoing edge from $u$ with label $c'$ and increment $\textit{Count}(z)$ by $|\range(p_1 c' p_2)|$. Let the heavy edge from $u$ have label $d$; we also increment $\textit{Count}(z')$, where $z'=\text{concat}(u, \text{concat}(v_{d}, w))$ is the node with path-label $p_1 d p_2$, by $|I_z|$ while processing node $z$.

This procedure then recurs on each of the side trees; i.e.~for side tree $S_i$, attached to node $u$, it calls \textsc{PerformCount}$(S_i,m-\mathcal{D}(u))$.
Finally, we construct array $C$ from array $\textit{Count}$ while performing one more depth-first traversal.

On the recursive calls of \textsc{PerformCount} in each of the side trees (e.g.~$S_i$) attached to $\textit{HP}$, we first compute the heavy paths (in $\cO(|S_i|)$ time for $S_i$) and then consider each node of string-depth $m$ of $\tr(x)$ at most once; as above, we process each node in $\cO(\log n)$ time due to Lemmas~\ref{suffix-of-node-u} and \ref{pcp}. As there are at most $n$ nodes of string-depth $m$, we do $\cO (n \log n)$ work in total. This is also the case as we go deeper in the tree. Since the number of leaves of the trees we are dealing with at least halves in each iteration, there at most $\cO(\log n)$ steps. Hence, each node of string-depth $m$ will be considered $\cO(\log n)$ times and every time we will do $\cO(\log n)$ work for it. The overall time complexity of the algorithm is thus $\cO (n \log^2 n)$. The space complexity is clearly $\cO(n)$. By applying Theorem~\ref{worstcase1} we obtain the following result.

\begin{theorem}\label{worstcase2}
Given a string $x$ of length $n$ over a fixed-sized alphabet and a positive integer $m$, we can solve the \textsc{1-mappability} problem in $\mathcal{O}(\min \{mn, n\log^2 n \})$ time and $\mathcal{O}(n)$ space.
\end{theorem}

\section{Final Remarks}\label{sec:conc}

The natural next aim is to either extend the presented solutions to work for arbitrary $k$ without increasing the time and space complexities dramatically or develop fundamentally new algorithms if this is not possible. In fact, we already know that the fast average-case algorithm presented in Section~\ref{sec:fast} can be generalized to work for arbitrary $k$ in linear time. This adds, however, a multiplicative factor of $k$ on the condition for the value of $m$. An interesting generalization of this problem would be to consider the edit distance model instead of the Hamming distance model; i.e.~apart from mismatches to also allow for letter insertions and deletions.

Furthermore, a practical extension of the aforementioned problem is the following. Given reads from a particular sequencing machine, the basic strategy for genome re-sequencing is to map a {\em seed} of each read in the genome and then try and {\em extend} this match~\cite{BLAST}. In practice, a seed could be for example the first $32$ letters of the read---the accuracy is higher in the prefix of the read. It is reasonable to allow for a few (e.g.~$k=2$) errors when matching the seed to the reference genome to account for sequencing errors and genetic variation. Hence a closely-related problem to genome mappability that arises naturally from this application is the following: {\em What is the minimal value of $m$ that forces $\alpha\%$ of starting positions in the reference genome to have $k$-mappability equal to $0$?} 

A standard implementation of the algorithm presented in Section~\ref{sec:fast},
when applied to a large sequence of length $n$, requires more than $20n$ bytes of internal memory. Such memory requirements are a significant hurdle to the mappability computation for large datasets on standard workstations. Another direction of practical interest is thus to devise efficient algorithms for the problems of $1$-mappability and $k$-mappability for the External Memory model of computation. 
Efficient algorithms for computing the suffix array and the longest common prefix array in this model are already known and shown to perform well in practical terms (see~\cite{SA-EM}, for example). Since the average-case algorithm in Section~\ref{sec:fast} scans the longest common prefix array from left to right sequentially, it would be interesting to see whether it can be implemented efficiently in external memory.

\bibliographystyle{plainurl}
\bibliography{reference}

\end{document}

%% file: __arxiv_preamble.tex
\usepackage{fullpage}
\usepackage{amsthm,amssymb,amsmath}  
\usepackage{xspace,enumerate}
\usepackage[capitalise]{cleveref}
\usepackage[utf8]{inputenc}
\usepackage{thmtools}
\usepackage{thm-restate}
\usepackage{authblk}

  \theoremstyle{plain}
  \newtheorem{theorem}{Theorem}
  \newtheorem{lemma}{Lemma}  
    
  \newtheorem{fact}{Fact}
  \newtheorem{observation}{Observation}
  \theoremstyle{definition}
  \newtheorem{definition}{Definition}
  \newtheorem{example}{Example}

\title{Faster algorithms for 1-mappability of a sequence}

\author[1]{Mai Alzamel}
\author[1]{Panagiotis Charalampopoulos}
\author[1]{Costas S. Iliopoulos}
\author[1]{Solon P. Pissis}
\author[1,2]{Jakub Radoszewski}
\author[3]{Wing-Kin Sung}

\affil[1]{
Department of Informatics, King's College London, London, UK\\
\texttt{[mai.alzamel,panagiotis.charalampopoulos,costas.iliopoulos,solon.pissis]@kcl.ac.uk}
}

\affil[2]{
Institute of Informatics, University of Warsaw, Warsaw, Poland\\
\texttt{jrad@mimuw.edu.pl}
}

\affil[3]{
Department of Computer Science, National University of Singapore, Singapore\\
\texttt{ksung@comp.nus.edu.sg}
}

\date{\vspace{-5ex}}
